\documentclass[a4paper, 9pt, twosides]{amsart}

\usepackage{lmodern}
\usepackage{amssymb}
\usepackage{amsthm}
\usepackage{amsaddr}
\usepackage{mathtools}
\usepackage{MnSymbol}
\usepackage{graphicx}
\usepackage{caption}
\usepackage{subcaption}
\usepackage{url}

\usepackage{enumitem}
\setitemize{nolistsep}
\setenumerate{nolistsep}
\usepackage{todonotes}
\usepackage{algorithm,algorithmic}

\allowdisplaybreaks
\newtheorem{theorem}{Theorem}
\newtheorem{corollary}{Corollary}
\newtheorem{lemma}{Lemma}

\newtheorem{defn}{Definition}

\newtheorem{prob}{Problem}
\newtheorem{remark}{Remark}
\newtheorem{example}{Example}

\newcommand{\abs}[1]{\left\lvert{#1}\right\rvert}

\newcommand{\pmat}[1]{\begin{pmatrix}#1\end{pmatrix}}

\newcommand{\R}{\mathbb{R}}

\newcommand{\D}{\mathcal{D}}

\newcommand{\last}{\textit{Last}}
\newcommand{\run}{\textit{Run}}
\newcommand{\Output}{\textit{Output}}

\newcommand{\exec}{\textit{Exec}}
\newcommand{\f}{fault}
\newcommand{\g}{ideal}
\newcommand{\A}{\mathcal{A}}
\newcommand{\Exec}{\text{Exec}}
\newcommand{\GM}{\text{GetMatrix}}

\newcommand{\QD}{Q_{\D}}
\newcommand{\Qdash}{Q'}
\newcommand{\Q}{Q}
\newcommand{\T}{T}

\allowdisplaybreaks

\title[]{Learning event-driven switched linear systems}
\author{Atreyee Kundu and Pavithra Prabhakar}
\thanks{Atreyee Kundu is with the Department of Electrical Engineering, Indian Institute of Science Bangalore, India, E-mail: atreyeek@iisc.ac.in. Pavithra Prabhakar is with the Department of Computer Science, Kansas State University, USA, E-mail: {pprabhakar@ksu.edu}}


\date{\today}

\begin{document}
    \begin{abstract}
	   We propose an automata theoretic learning algorithm for the identification of black-box switched linear systems whose switching logics are event-driven. A switched system is expressed by a deterministic finite automaton (FA) whose node labels are the subsystem matrices. With information about the dimensions of the matrices and the set of events, and with access to two oracles, that can simulate the system on a given input, and provide counter-examples when given an incorrect hypothesis automaton, we provide an algorithm that outputs the unknown FA. Our algorithm first uses the oracle to obtain the node labels of the system run on a given input sequence of events, and then extends Angluin's \(L^*\)-algorithm to determine the FA that accepts the language of the given FA. We demonstrate the performance of our learning algorithm on a set of benchmark examples.
    \end{abstract}

\maketitle

    \section{Introduction}
\label{s:intro}
    Cyber-physical systems that consist of software controlled physical systems have transformed today's transportation, energy and healthcare sectors. Rigorous analysis of these systems has become inevitable given the safety critical environments in which they are deployed. Formal analysis requires a formal model of the system to be analyzed. Often, a model of the system is unavailable, due to, for instance, unknown dynamics or proprietary software, or a complex model maybe available, which is unamenable to analysis. In either case, it is necessary to have techniques to learn such models, from minimalistic knowledge of the system, and some basic operations that are feasible as in a black box setting. In this paper, we investigate the learning problem for certain subclasses of models for cyber-physical systems, wherein, the digital logic (cyber part) is captured as a event-driven deterministic finite state automaton, and the physical system is captured using discrete-time linear dynamics.

    In this paper, we focus on \emph{event-driven switched linear systems}. In general, switched systems consist of a finite set of subsystems governed by a time-varying switching signal \cite[\S 1.1.2]{Liberzon}. We focus on linear subsystems and sets of switching logic that are event-driven in the sense that the active subsystem at any time instant depends on the active subsystem at the previous time instant and the event that was carried out at that time instant. Such switched systems arise naturally in, for example, the setting of a robot-aided neurosurgery \cite{Comparetti2014}. 
    Consider, for instance, a robot that has five modes of operation: (i) Homing, (ii) Autonomous, (iii) Hands-on, (iv) Tele-operation, and (v) Steady. In each mode, it has a certain dynamics, and the mode can change based on certain events. The surgeon is provided with a GUI interface where she can perform an event by pressing a button or touching the robot. Based on the current mode of operation of the robot (subsystem) and the event carried out by the surgeon, the next mode of operation of the robot (subsystem) is selected.  To have successful coordination between the human and the robot, it is imperative to understand the functioning of the robot. 
    Hence, we are interested in developing system identification techniques for event-driven switched linear systems by employing automata theoretic learning techniques.
    Note that we can assume that we have the ability to stimulate the robot with input event sequences, and observe its behavior (execution). Our algorithm allows one to compute a "hypothesis" switched system based on such observations on appropriate input event sequences. In addition, if we have the ability to check if the hypothesized system is correct, and obtain a counter-example execution in case it is incorrect, our algorithm can learn the correct system in finite time.

    We express an event-driven switched system as an event-deterministic finite automaton (FA), whose node labels are the subsystem matrices. The execution of a switched system depends on an initial (continuous) state and a sequence of input events, and consists of the sequence of states obtain by applying the discrete-time linear dynamics associated with (discrete) state labels that are encountered along the path in the finite automaton induced by the input event sequence.  We assume that the set of events that causes switches between the subsystems and the dimension of the subsystems matrices are known to the Learner. In addition, she has access to two Oracles:
    \begin{itemize}[label = \(\circ\), leftmargin = *]
        \item An IO-generator, which given an initial state and an input (sequence of events) outputs the execution. Such a IO-generator is typically available for any black box for which an input can be provided and output observed. 
        \item An Equivalence Checker, which given a hypothesis finite automaton, checks if the language of the hypothesized finite automaton is the same as that underlying the black box automaton, and if they are not equal, provides an input on which the two automata have different outputs. While such an equivalence checker/counter-example generator might be challenging to build, it might still be possible to generate counter-examples by running multiple IO-generator queries and observing the output.
    \end{itemize}

    Under the above assumptions, our learning algorithm has the following phases:
    \begin{itemize}[label = \(\circ\), leftmargin = *]
        \item First, the Learner performs a constant number of IO-generator queries to obtain the matrix labeling the last node of the automaton path corresponding to an input, referred to as the \emph{$\Output$}.
        \item The switched system learning problem is then reduced to a finite automaton learning problem with \emph{multiple} labels. We provide an extension of Angluin's \(L^*\)-algorithm \cite{Angluin1987} to the multiple labels setting, using the notion of $\Output$ as our observation.
    \end{itemize}
    The key insight of our algorithm is that we are able to separate the learning tasks into a dynamics identification task and an automata learning tasks.
    We are able to provide guarantees that our learning algorithm terminates in bounded time and outputs a correct language equivalent switched system.  Our algorithm is tested on a set of benchmark examples. 
    
    The remainder of this paper is organized as follows: We present a discussion on existing techniques for both switched systems identification and automaton learning in \S\ref{s:prior_works}. In \S\ref{s:prob_stat} we formulate the problem under consideration. Our results appear in \S\ref{s:res}. We also discuss various features of our learning algorithm in this section. A set of numerical experiments is presented in \S\ref{s:num_ex}. We conclude in \S\ref{s:concln} with a brief mention of future research directions.



    \section{Related Work}
\label{s:prior_works}
In this section, we provide a brief overview of related work in the area of system identification and automata based learning.

\subsection{Systems identification techniques for switched systems}
\label{ss:lit_survey}
    The knowledge of mathematical models of the subsystems (e.g., transfer functions, state-space models, or kernel representations) and restrictions on the set of admissible switching signals are key requirements for the design of decision and control algorithms for switched systems. As a result, system identification techniques for these systems are widely studied, see e.g., the survey paper \cite{Garulli2012}, the tutorial paper \cite{Paoletti2007} and the references therein.

    In general, the problem of identification of switched systems is known to be NP-hard \cite{Lauer2016} and is typically performed by collecting input (possibly controlled) - output (possibly noisy) data during the operations of the system. The available techniques can be classified broadly into two categories: (A) offline methods and (B) online methods. In case of the former, access to all data at once is assumed, while in case of the latter, data are available in a streaming (online) fashion. The offline methods include:
    (a) Algebraic method \cite{Vidal2003} that uses Veronese embedding to decouple the tasks of estimating the subsystems parameters and the switching signals. An exact solution is obtained when the subsystems evolution and the available data are noise-free. This technique is extended to the setting where subsystems evolutions and/or the available data are noisy in \cite{Ozay2015}. The authors convert the algebraic method to a rank minimization problem that is solved by employing a semi-definite program.
    (b) Mixed integer programming method \cite{Roll2004} that involves linear or quadratic programming techniques whose solutions are shown to converge to global optima. The proposed set of algorithms is particularly useful in the settings where obtaining data is an expensive process and relatively few data are available.
    (c) Clustering method \cite{Ferrari-Trecate2003} that combines clustering, linear identification and pattern recognition techniques. The identification of the subsystems and the state-space regions in which they are active is performed by avoiding a commonly used gridding technique. In addition, the available data are classified in a carefully designed feature space that allows reconstruction of different subsystems that have similar parameters but operate on different regions.
    (d) Bayesian method \cite{Juloski2005} that treats the subsystems parameters as random variables described by their probability density functions. The identification process involves computation of a posteriori probability density functions of the subsystems parameters and employs the information derived in the previous steps of the identification process for estimating the state-space regions in which various subsystems are active.
    (e) Bounded error identification method \cite{Bemporad2005} that first classifies the available data and obtains estimates of the number of subsystems and parameters of the subsystems by solving a set of linear inequalities, and then employs a refinement procedure to reduce misclassifications. An upper bound on the identification error is maintained as a tuning parameter at all times during the identification process.
    (f) Sparse optimization method \cite{Bako2011} that poses the identification problem as an NP-hard combinatorial \(\ell_0\) optimization problem. Sufficient conditions for solving it are presented by employing relaxations to convex \(\ell_1\)-norm minimization problems from compressed sensing literature. It is demonstrated that a priori clustering of the available data corresponding to the various subsystems is not a necessary step for system identification. The online methods, receive data at each time step and perform two tasks: identification of the subsystem whose dynamics is being followed at that time step and updation of the estimates of the subsystems parameters. In \cite{Vidal2008} the author studies online identification of switched systems as an extension of the offline algebraic method (see (a) above). The works \cite{Bako2011_1,Goudjil2016,Du2018} employ two-step procedures for online identification of switched systems. First, candidate estimates for each subsystem are built, and second, at every time, the active subsystem is determined by assigning the data to one of the candidates according to some criteria and the estimates of the candidates are updated. In particular, \cite{Bako2011_1} employs prior or posterior residual error for the identification of active subsystems and recursive least squares for updating the candidate estimates, while \cite{Goudjil2016} employs minimization of prior residual error for the identification of active subsystems and a modified outer bounding ellipsoid algorithm for the updation of candidate estimates. The residual error approach for the identification of active subsystem at every time step is modified to a robust version by incorporating an upper bound on estimation error in \cite{Du2018}. The authors employ a randomized Kaczmaz algorithm and normalized least mean squares towards updating the candidate estimates of the subsystems parameters. In this paper we consider a paradigm shift and explore active learning techniques for system identification of switched systems.
\subsection{Automata learning techniques}
\label{ss:automat_lang}
The \(L^*\)-algorithm learns a minimal deterministic finite automaton that accepts a certain regular language by employing two types of queries: membership query and equivalence query. A Teacher aids the learning process by answering whether a given string is in the language as well as whether an automaton hypothesized by the Learner is correct or not. \(L^*\) is an online learning algorithm in the sense that the Learner is allowed to ask further queries and enlarge her database as and when needed. This algorithm is extended to the learning of non-deterministic finite automaton in \cite{Bollig2009}, probabilistic finite automaton in \cite{Tzeng1992}, oracle automaton for software testing in \cite{Barr2015}, input-output automaton in \cite{Aarts2010}, register automata in \cite{Howar2012,Cassel2016} and Moore machines with decomposable outputs in \cite{Moerman2019}. In this paper we extend \(L^*\)-algorithm to learn event-driven deterministic finite automaton whose nodes are labelled with matrices. 
    In general, automata learning techniques are employed widely in model learning \cite{Groz2020,Aichernig2018}, model checking \cite{Clarke1999}, automatic verification of networks of processes \cite{Grinchtein2006}, compositional verification \cite{Cobleigh2003,Giannakopoulou2013}, as well as conformance testing of boolean programs \cite{Kumar2006}. In this paper we employ the \(L^*\)-algorithm proposed by Dana Angluin in \cite{Angluin1987} as a primary tool for our learning task.


    \section{Problem statement}
\label{s:prob_stat}
In this section we present the mathematical formulation of the learning problem under consideration. We define the necessary preliminaries, the class of systems we intend to learn, and the assumptions for our learning problem.

\subsection{Notation}
\label{ss:notation}
    \(\R\) will denote the set of real numbers, \(I_{d}\) the \(d\)-dimensional identity matrix and \(I_{d}^{k}\) its \(k\)-th column.
    For a finite set \(A\), its cardinality is denoted by \(\abs{A}\).
    A (finite) sequence over a set \(A\) is denoted by listing elements from \(A\), e.g., \(w = a_1 a_2\cdots a_n\), where \(a_i\in A\), \(i=1,2,\ldots,n\). \(\varepsilon\) denotes an empty sequence. We employ \(\last(w)\) to denote the last element of the sequence \(w\), i.e., \(\last(w) = w_n\).
    Also, $w[i\cdots j]$ represents the sequence $a_i \cdots a_j$.
    Let \(A^*\) denote the set of all finite sequences over \(A\).


\subsection{Switched systems}
In this paper, we study learning algorithms for a subclass of discrete-time linear switched systems, wherein the switching logic is specified by a finite automaton.
We first define a finite automaton and its language.

\label{ss:swsys}
    \begin{defn}
    \label{d:automata}
    \rm{
        An event-deterministic labelled finite automaton (FA) is a tuple \(\D = (\Q,q_{0},\Sigma,\Lambda, \delta, \gamma)\), where \(\Q\) is the set of nodes or (discrete) states, \(q_{0}\in \Q\) is the initial node, \(\Sigma\) is a set of events, \(\Lambda\) is a set of node labels, \(\delta:\Q\times\Sigma\to \Q\) is the node transition function, and \(\gamma:\Q\to\Lambda\) is the node labelling function.
    }
    \end{defn}

In the sequel, we will refer to an event-deterministic labelled finite automaton, as just a finite automaton. The components of a finite automaton will be identified by using subscripts indicating the automaton, such as, $\D$ will refer to the nodes of automaton $\QD$. When the automaton is clear from the context, the subscripts will be dropped.

    \begin{example}
    \label{ex:fa}
    \rm{
        Consider an FA $\D = (\Q,q_{0},\Sigma,\Lambda, \delta, \gamma)$ shown in Figure \ref{fig:DLFA1}. It has \(\Q = \{q_{0},q_{1},q_{2},q_{3}\}\), \(\Sigma = \{e_{1},e_{2}\}\), \(\Lambda = \{\ell_{1},\ell_{2},\ell_{3}\}\), \(\delta(q_{0},e_{1}) = q_{3}\), \(\delta(q_{0},e_{2}) = q_{1}\),
			\(\delta(q_{1},e_{1}) = q_{2}\), \(\delta(q_{1},e_{2}) = q_{0}\),
			\(\delta(q_{2},e_{1}) = q_{1}\), \(\delta(q_{2},e_{2}) = q_{3}\),
             \(\delta(q_{3},e_{1}) = q_{0}\), \(\delta(q_{3},e_{2}) = q_{2}\), and
            \(\gamma(q_{0}) = \ell_{1}\), \(\gamma(q_{1}) = \ell_{2}\), \(\gamma(q_{2}) = \ell_{2}\), \(\gamma(q_{3}) = \ell_{3}\).
	    \begin{figure}[htbp]
    \centering
        \scalebox{0.6}{
        \begin{tikzpicture}[every path/.style={>=latex},base node/.style={draw,circle}]
            \node[base node] (a) at (-2,2)  { $\ell_{1}$ };
            \node[base node] (b) at (0,2)  { $\ell_{2}$ };
            \node[base node] (d) at (-2,-2)  { $\ell_{3}$ };
            \node[base node] (c) at (0,-2)  { $\ell_{2}$ };

            \draw[->] (a) edge (b);
            \draw[->] (b) edge[bend right] (a);
            \draw[->] (b) edge (c);
            \draw[->] (c) edge[bend right] (b);
            \draw[->] (c) edge (d);
            \draw[->] (d) edge[bend right] (c);
            \draw[->] (a) edge (d);
            \draw[->] (d) edge[bend left] (a);

             \node (s) at (1,0.2) {$e_{1}$};
             \node (s) at (-0.2,0.2) {$e_{1}$};
             \node (s) at (-1.8,0.2) {$e_{1}$};
             \node (s) at (-2.9,0.2) {$e_{1}$};

             \node (s) at (-1,1.8) {$e_{2}$};
             \node (s) at (-1,2.6) {$e_{2}$};
             \node (s) at (-1,-2.6) {$e_{2}$};
             \node (s) at (-1,-1.8) {$e_{2}$};

            \draw[->] (-3,2) -- (a);
        \end{tikzpicture}
        }
        \caption{FA for Example \ref{ex:fa}} \label{fig:DLFA1}
    \end{figure}
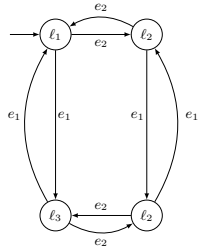
    }	
    \end{example}

Note that the transition function of our automaton is deterministic.
We will refer to a sequence of event, that is, an element of $\Sigma^*$, as an input (word or sequence).
We overload \(\delta\) to also denote the function \(\delta:\Q\times\Sigma^*\to \Q\) that given a state and an input word and outputs the state reached on taking the sequence of transitions corresponding to the input word, and is inductively defined as \(\delta({q},\varepsilon) = {q}\) and \(\delta({q},ua) = \delta(\delta({q},u),a)\) for all \(u\in\Sigma^*\) and \(a\in\Sigma\). Similarly, we overload \(\gamma\) to a function \(\gamma:\Q^*\to\Lambda^*\) given by  \(\gamma({q}_{0} {q}_{1}\ldots {q}_{n}) = \gamma({q}_{0})\gamma({q}_{1})\ldots\gamma({q}_{n})\) for all \({q}\in \Q^*\).

We will define the semantics of an FA as a mapping from input words to corresponding sequence of state labels generated by them. We will refer to this mapping as a "language".
First, we define a run of the FA on a word to be the sequence of nodes generated by reading the word.
    \begin{defn}
    \label{d:run}
    \rm{
        Given \(w = {e}_1 {e}_2\cdots {e}_n\in\Sigma^*\), \emph{run of \(w\) on \(\D\)} is given by
        \begin{align}
        \label{e:run}
            \run_{\D}(w) = {q}_{0} {q}_{1}\cdots {q}_{n}\:\:\text{for any}\:w\in\Sigma^*,
        \end{align}
        where \({q}_{i+1} = \delta({q}_{i},{e}_{i+1})\), for \(i=0,1,\ldots,n-1\).
    }
    \end{defn}

    \begin{defn}
    \label{d:semantics}
    \rm{
        The \emph{language of \(\D\)} is a function \(L_{\D}:\Sigma^*\to\Lambda^*\) given by
        \begin{align}
        \label{e:language}
            L_{\D}(w) = \gamma(\run_\D(w)).
        \end{align}
        }
    \end{defn}

    In the sequel, for learning, we will need the label of the last node reached on reading a word. We will refer to this as the output. This is a generalization of the notion of acceptance of a word by a traditional deterministic finite automaton, where the labels are "final" and "non-final".
    \begin{defn}
    \label{d:output}
    \rm{
        Given \(w = {e}_1 {e}_2\ldots {e}_n\in\Sigma^*\), the \emph{output of \(w\) in \(\D\)} is the label of the last node of \(\run(w)\). More specifically,
        \begin{align}
        \label{e:output}
            \Output_{\D}(w) = \last(L_{\D}(w)).
        \end{align}
        }
    \end{defn}

    \begin{example}
    \label{ex:fa_parameters}
    \rm{
        Consider the FA in Example \ref{ex:fa}. Let \(w = e_1 e_2 e_1 e_2 e_2\). We have \(\run(w) = q_{0} q_{3} q_{2} q_{1} q_{0} q_{1}\), \(L_{\D}(w) = \ell_1 \ell_3 \ell_2 \ell_2 \ell_1 \ell_2\), and \(\Output_{\D}(w) = \ell_2\).
        }
    \end{example}

 We consider switched systems consisting of a finite number of discrete-time dynamical systems, each of which is specified by a matrix $A_i$, with the intended dynamics being $x(t+1) = A_i x(t)$,  and a switching logic specified using a finite automaton. We capture the switched system holistically as a finite automaton with the matrices being the node labels.
    \begin{defn}
    \label{d:swsys}
    \rm{
        A \emph{switched system} is a FA \(\D\), whose set of node labels, \(\Lambda_{\D}\), is an indexed set of matrices of dimension \(d\) represented as \(\Lambda_{\D} = \{A_j\}_{j=1}^{N}\), where \(A_j \in \R^{d\times d}\) for every $j$.
        }
    \end{defn}
    In the sequel, we will occasionally refer to the elements of the set \(\{A_j\}_{j=1}^{N}\) as subsystem matrices.
    An execution of $\D$ from an initial (continuous) state \(x\in\R^{d}\) on an input sequence of events  $w$, denoted $\exec_\D(x, w)$, is the sequence of states reached by applying the dynamics represented by the matrices labelling the nodes in the run of the finite automaton on the input sequence.
    In the sequel, we will need a general definition of executions from a finite number of, say, $k$ initial states, stored as a $d \times k$-dimensional matrix, each of whose columns represents a state. The execution will be a sequence of $d \times k$-dimensional matrices, where the $i$-th column of these matrices represents the execution starting from the $i$-th column of the initial matrix.

    \begin{defn}
    \label{d:swsys_exec}
    \rm{
        An \emph{execution of \(\D\)}, on a state matrix \(X \in\R^{d \times d}\), and a sequence of events, \(w = {e}_1 {e}_2\cdots {e}_n\in\Sigma^*\), is given by
        \begin{align}
        \label{e:swsys_exec}
            \exec_{\D}(X,w) = X_0 X_1 \cdots X_{n+1},
        \end{align}
        where \(X_0 = X\), \(X_{i+1} = \overline{A}_i X_i\), \(i=0,1,\ldots,n\), and \(L_{\D}(w) = \overline{A}_0 \overline{A}_1\ldots \overline{A}_n\).
        }
    \end{defn}

Note that given a state $x \in\R^d$, $\exec_\D(x, w)$ denotes the execution from a $d \times 1$ matrix.
We use states to refer to both elements of $\R^d$, which are continuous states, and nodes in $\Q$, which are discrete states. When there is ambiguity, we will use the prefix "discrete"/"continuous".

    \begin{example}
    \label{ex:swsys}
    \rm{
        Consider a switched system given by the FA \(\D \) described in Example \ref{ex:fa}, with \(\ell_1 = A_1 = \pmat{1.0 & 0.3\\0.7 & 1.2}\), \(\ell_2 = A_2 = \pmat{0.4 & 0.8\\-0.7 & 0.6}\) and \(\ell_3 = A_3 = \pmat{1.2 & 0.7\\1.6 & 0.1}\). Let \(x = \pmat{0.5\\0.5}\) and \(w = e_1 e_2 e_1 e_2 e_2\). Then
        \begin{align*}
            \exec_{\D}(x,w) &= x_0x_1x_2x_3x_4x_5x_6\\
            &= x_0\: A_1x_0\: A_3x_1\: A_2 x_2\: A_2 x_3\: A_1 x_4\: A_2 x_5\\
            &= \pmat{0.5\\0.5}\pmat{0.65\\0.95}\pmat{1.445\\1.135}\pmat{1.486\\-0.3305}\pmat{0.33\\-1.2385}\\
            &\quad\quad\pmat{-0.04155\\-1.2552}\pmat{-1.02078\\-0.724035}.
        \end{align*}
        }
    \end{example}

As before, when the finite automaton or the switched system is clear from the context, we will drop the subscript $\D$ from $\Output$,  $\run$,  $L$ and  $\exec$.

\subsection{Learning problem}
    Our broad objective is to learn a switched system, which is provided as a black box system.  We now formalize our learning problem.
    \begin{prob}
    \label{prob:main}
     Consider a switched system \(\D= (\QD,q_{0,\D},\Sigma,\Lambda_{\D},\delta_{\D},\gamma_{\D})\). Suppose that we know the set of events, \(\Sigma\), and the dimension, \(d\), of the elements of \(\Lambda_{\D}\). In addition, we have access to two oracles that can perform the following tasks:
         \begin{enumerate}[label = (\alph*),leftmargin = *]
            \item IO-generator: Given input \((x,w)\in\R^{d}\times\Sigma^*\), the IO-generator outputs \(\exec_{\D}(x,w)\). Note that we can find $\exec_\D(X, w)$ for any $d \times k$-matrix by $k$ calls to the IO-generator.
            \item Equivalence-checker (counter-example generator): Given a (hypothesis) FA \(\D' = (\Qdash,q'_{0},\Sigma,\Lambda',\delta',\gamma')\) as input, the equivalence checker checks the equivalence of the languages of $\D$ and $\D'$, that is, it outputs if $L_\D = L_{\D'}$. If not, then it also outputs a (counter-example) \(w\in\Sigma^*\) such that \(\Output_{\D}(w)\neq \Output_{\D'}(w)\).
         \end{enumerate}
    Our objective is to design an algorithm that uses the above oracles to output an automaton $\D'$ such that \(L_{\D} = L_{\D'}\).
    \end{prob}

    In the sequel, we will also refer to a call to IO-generator on an input word and a continuous state for obtaining an execution of the black box switched system, as an \emph{observation query}, and the call to the equivalence checker with a hypothesis automaton, an \emph{equivalence query}. Towards solving Problem \ref{prob:main}, we will assume that the matrices \(\{A_j\}_{j=1}^{N}\) are full-rank, and devise a learning algorithm that relies on the principles of Angluin's $L^*$ algorithm.
    Our solution approach broadly consists of the following steps:
    \begin{itemize}[label = \(\circ\), leftmargin = *]
        \item
         We use the IO-generator to compute $\Output_\D(w)$ for a given $w$, thereby reducing the learning problem to that of learning an event-deterministic labelled finite automaton.
         \item
         We extend the \(L^*\)-algorithm for deterministic finite automata (with two labels, namely, final and non-final) to the setting of learning event-deterministic finite automata with potentially multiple labels.
    \end{itemize}

    \section{Switched System Learning Algorithm}
\label{s:res}
    This section contains the details of our solution to Problem \ref{prob:main}. We begin with an algorithm to compute \(\Output_{\D}(w)\) for a given $w$ by making a sequence of IO-generator queries that provide $\exec_\D(x,w')$ as output for a given initial state $x$ and input $w'$.
    Then we provide an algorithm that learns the underlying finite automaton that has access to the equivalence checker and the algorithm for computing $\Output_\D(\cdot)$.

\subsection{Computation of \(\Output_{\D}(w)\)}
\label{ss:res_set1}
The computation of $\Output_\D$ relies on the fact that a matrix $A$ can be uniquely computed given a set of basis vectors and their transformation on the application of $A$, when $A$ is full-rank. Let \(\GM\) be a function that takes as input a matrix $X$ whose columns form a basis, and the transformation of those vectors on a matrix $A$, given by a matrix $X' = AX$, and returns $A$. More precisely, $\GM(X, X')$ takes as input two matrices \(X,X'\in\R^{d\times d}\) whose columns form a basis, and solves the systems of linear equations \(AX = X'\) for \(A\in\R^{d\times d}\), and returns \(A\). Such a matrix can be constructed effectively by solving the system of linear equations, and the uniqueness of the solution is guaranteed by well-known results from linear algebra.

To obtain $\Output_\D(w)$, we need to find two sets of basis vectors, where the second one corresponds to a transformation of the first using the matrix $\last(\gamma_\D(\run_\D(w)))$.
The algorithm is quite straight forward.
Consider $I_d$, a $d \times d$ identity matrix, whose columns form a basis.
Let $X = \last(\Exec_\D(I_d, w[1 \cdots n-1]))$, where $w$ is a sequence of $n$ events.
Note that the columns of matrix $X$ also form a basis, because all the matrices in $\gamma_\D(\run_\D(w))$ are full rank matrices.
Similarly, let $X' = \last(\Exec_\D(I_d, w[1 \cdots n]))$, which again represents a basis.
Moreover, $X' = \Output_\D(w) X$. Hence, $\Output_\D(w)$ is given by $\GM(X, X')$.
This construction of $\Output_\D(w)$ is outlined in Algorithm \ref{algo:output_comp}.

	%

	\begin{algorithm}[htbp]
			\caption{Computation of \(\Output_{\D}(w)\)} \label{algo:output_comp}
		\begin{algorithmic}[1]
			\renewcommand{\algorithmicrequire}{\textbf{Input:}}
			\renewcommand{\algorithmicensure}{\textbf{Output:}}
			
			\REQUIRE The dimension of the subsystems matrices, \(d\) and a sequence of events, \(w \in\Sigma^*\).	
            \ENSURE \(\Output_{\D}(w)\).

            \IF {\(w = \varepsilon\)}
                \STATE Output $\Exec_{\D}(I_d,\varepsilon)$ and terminate.
            \ELSE
                    \STATE Set \(X = \last(\Exec_{\D}(I_d,w[1\cdots n-1]))\)\\
                    \STATE Set \(X' = \last(\Exec_{\D}(I_d,w[1\cdots n]))\)\\
                    \STATE Output $\GM(X, X')$ and terminate.
            \ENDIF
        \end{algorithmic}
	\end{algorithm}

    \begin{lemma}
     \label{lem:auxres1}
         Given \(w\in\Sigma^*\), Algorithm \ref{algo:output_comp} outputs \(\Output_{\D}(w)\).
    \end{lemma}
    We now present an example to demonstrate Algorithm \ref{algo:output_comp}.

    \begin{example}
    \label{ex:output_comp}
    \rm{
        Recall Example \ref{ex:swsys}. We apply Algorithm \ref{algo:output_comp} to compute \(\Output_{\D}(w)\) for \(w = e_1 e_2\). The following steps are carried out:
        \begin{enumerate}[label = \arabic*), leftmargin = *]
           \item i) Input \(\biggl(\pmat{1\\0},e_1\biggr)\) to the IO-generator and observe \(\exec_{\D}\biggl(\pmat{1\\0},e_1\biggr)=\pmat{1\\0},\pmat{1.0\\0.7},\pmat{1.69\\1.67}\).\\
           ii) Input \(\biggl(\pmat{0\\1},e_1\biggr)\) to the IO-generator and observe \(\exec_{\D}\biggl(\pmat{0\\1},e_1\biggr)=\pmat{0\\1},\pmat{0.3\\1.2},\pmat{1.2\\0.6}\).
           \item i) Input \(\biggl(\pmat{1\\0},e_1e_2\biggr)\) to the IO-generator and observe \(\exec_{\D}\biggl(\pmat{1\\0},e_1e_2\biggr)=\pmat{1\\0},\pmat{1.0\\0.7}\),\\\(\pmat{1.69\\1.67}\),\(\pmat{2.012\\-0.181}\).\\
           ii) Input \(\biggl(\pmat{0\\1},e_1e_2\biggr)\) to the IO-generator and observe \(\exec_{\D}\biggl(\pmat{0\\1},e_1e_2\biggr)=\pmat{0\\1},\pmat{0.3\\1.2}\),\\\(\pmat{1.2\\0.6}\),\(\pmat{0.96\\-0.48}\).
           \item We have i) \(X = \exec_{\D}(I_d,e_1) = \pmat{1.69 & 1.2\\1.67 & 0.6}\), and
           ii) \(X' = \exec_{\D}(I_d,e_1 e_2) = \pmat{2.012 & 0.96\\-0.181 & -0.48}\).
           \item We solve the systems of linear equations \(X' = AX\) for \(A = \pmat{a_{11} & a_{12}\\a_{21} & a_{22}}\),
                and obtain \(a_{11} = 0.4\), \(a_{21} = -0.7\), \(a_{12} = 0.8\) and \(a_{22} = 0.6\).
        \end{enumerate}
        }
    \end{example}

    Armed with Algorithm \ref{algo:output_comp}, we proceed towards extending \(L^*\)-algorithm from the learning literature to the learning of a FA \(\D^*\) that accepts the language of \(\D\).

\subsection{Learning algorithm}
\label{ss:res_set2}
Let us fix an unknown finite automaton $\D$, for which we know the set of events $\Sigma$ and the dimension of the matrices in $\Lambda_\D$. Our objective is to output a finite automaton $\D'$ such that $L_\D = L_{\D'}$. We have access to an algorithm for computing $\Output_\D(w)$ for any given input $w$, from Algorithm \ref{algo:output_comp}.

    The broad framework of our learning approach based on Angluin's $L^*$ algorithm is as follows: at each step of the learning algorithm, the Learner maintains two sets of input words (sequences over $\Sigma$): $\Q$, a set of access words, and $\T$, a set of test words.
    Intuitively, the set $\Q$ represents a set of input words that reach distinct states in any minimal finite automaton $\D^*$ representing the language to be learnt.
    Note that for any two distinct states of $\D^*$, there is an input word, that will distinguish the behaviors from those states.
    $\T$ is a finite set of input words that can distinguish any pair of states in $\Q$.
    This property is referred to as $(\Q, \T)$ being $\D$-separable.
    The algorithm consists of judiciously expanding $\Q$ and when required $\T$, so that $(\Q, \T)$ separability is maintained.
    In each step, a hypothesis automaton is constructed from $\Q$ by possibly adding states to "close" the automaton, that is, to ensure that there is a next state on every event from every state. The language of the closed automaton is compared with $\D$ using an equivalence query, and a counter-example if returned, is used to identify a state that has not been captured by the set $\Q$ and added. The process is repeated until a finite automaton which passes the equivalence query is found.


First, we define when two input words are equivalent with respect to a set of test words $\T$.
    \begin{defn}
    \label{d:T-equivalence}
    \rm{
        Given a set \(T \subseteq \Sigma^*\), and two words \(u, v\in\Sigma^*\), we say that $u, v$ are \emph{\(\T\)-equivalent} with respect to \(\D\), denoted by \(u\equiv_{\T}^{\D}v\), if
        \begin{align}
        \label{e:T-equiv}
            \Output_{\D}(uw) = \Output_{\D}(vw)\:\:\text{for all}\:w\in \T.
        \end{align}
        }
    \end{defn}
Given a finite $\T$ and input words $u, v$, we can algorithmically check if $u, v$ are $\T$-equivalent, by iterating over words $w \in \T$ and using Algorithm \ref{algo:output_comp} to check if $\Output_{\D}(uw) = \Output_{\D}(vw)$.
Note that if $u, v$ are not $\T$-equivalent, then some word $w$ from $\T$ distinguishes them, in terms of the label of the last state reached after reading $w$ from the states reached after reading $u$ and $v$, respectively.
This leads us to the notion of separability, which guarantees that the states reached by words in a set $\Q$ of access strings are distinct, using a finite set of test strings $\T$ that witness the distinguishability.
    \begin{defn}
    \label{d:separability}
    \rm{
        The pair \((\Q,\T)\) is called \emph{\(\D\)-separable}, if no two distinct words in \(\Q\) are \(\T\)-equivalent with respect to \(\D\).
        }
    \end{defn}
    Again, given that we can check if $u, v$ are $\T$-equivalent for a finite $\T$, we can also algorithmically check if $(\Q, \T)$ $\D$-separable, when $\Q$ is also finite.
    Given a set of access strings $\Q$ that reach distinct states of $\D$, we want to hypothesize a finite automaton that captures the language of $\D$. We need to identify the states and transitions of this automaton. We can consider $\Q$ to represent the  states, with the interpretation that they represent the states reached in $\D$ when given themselves as input.
    In order to define the edge, for every $q \in \Q$ and $e \in \Sigma$, we need to identify a word in $\Q$ that corresponds to $qe$. We can search for a word in $\Q$, that is $\T$-equivalent to $qe$. Note that there is at most one such word in $\Q$ if $(\Q, \T)$ is separable. However, no such word might exist. Hence, we add those words to $\Q$, until $\Q$ is closed with respect to the "next step" operation.
    Next, we formalize the notion of closure, and the hypothesis automaton constructed when a closed pair $(\Q, \T)$ is given.

    \begin{defn}
    \label{d:closure}
    \rm{
        The pair \((\Q,\T)\) is called \emph{\(\D\)-closed}, if for every \(q\in \Q\) and \(e\in\Sigma\), there exists \(q'\in \Q\) such that \(qe\equiv_{\T}^{\D}q'\).
        }
    \end{defn}

    \begin{defn}
    \label{d:hypo_automat}
    \rm{
        Consider a \(\D\)-separable and \(\D\)-closed pair \((\Q,\T)\). A hypothesis automaton for $(\Q, \T)$ is a finite automaton \(\D' = (\Qdash,q'_{0},\Sigma,\Lambda',\delta',\gamma')\), where:
        \begin{itemize}[label = \(\circ\), leftmargin = *]
            \item \(\Qdash=\Q\) with the empty sequence of events, \(\varepsilon\), being the initial node, that is, $q'_0 = \varepsilon$;
            \item $\Lambda' = \{\Output_\D(q) \,|\, q \in \Q\}$;
            \item for any $q, e$, \(\delta(q,e) = q'\), where $q' \in \Q$ is such that \(qe\equiv_{\T}^{\D}q'\);
            \item for any $q$, \(\gamma(q) = \Output_{\D}(q)\).
        \end{itemize}
    }
    \end{defn}
Note that our definition of hypothesis automaton is well-defined, since, in the definition of $\delta'$, there is at most one $q'\in \Q$ satisfying \(qe\equiv_{\T}^{\D}q'\), because of the separability property of $(\Q, \T)$.
Also, checking for whether a pair of finite sets $(\Q, \T)$ is closed and the construction of the hypothesis automaton for $(\Q, \T)$ are computable.

Our learning algorithm is summarized in Algorithm \ref{algo:semantics}. The details and correctness of the algorithm depend on the following results.
    \begin{algorithm}[htbp]
			\caption{Learning a minimal FA whose semantics is \(L_{\D}\)} \label{algo:semantics}
		\begin{algorithmic}[1]
			\renewcommand{\algorithmicrequire}{\textbf{Input:}}
			\renewcommand{\algorithmicensure}{\textbf{Output:}}
			
			\REQUIRE The set of events, \(\Sigma\) and the dimension of the subsystems, \(d\), Algorithm for computing $\Output_\D$ and Counter-example generator for the language $L_\D$
            \ENSURE A FA \(\D'\) whose language is \(L_{\D}\)

            \STATE Set \(\Q = \T = \{\varepsilon\}\).
            \STATE \label{line:2} Apply Lemma \ref{lem:auxres4a} to find \(\tilde{\Q} \supseteq \Q\) such that \((\tilde{\Q},\T)\) is \(\D\)-separable and \(\D\)-closed.
            \STATE Set \(\Q = \tilde{Q}\)
            \STATE Construct a hypothesis automaton, \(\D'\) for the pair \((\Q,\T)\)
            \STATE Check equivalence of $\D'$ and $\D$
            \IF {a counterexample \(w \in\Sigma^*\) is returned}
                \STATE Apply Lemma \ref{lem:auxres4} to expand \(\Q\) and \(\T\) towards obtaining a \(\D\)-separable pair \((\tilde{\Q},\tilde{\T})\)
                \STATE Set \(\Q = \tilde{Q}\) and \(\T=\tilde{T}\)
                \STATE Go to Line \ref{line:2}
            \ELSE
                \STATE Output \(\D'\) and terminate.
            \ENDIF
		\end{algorithmic}
	\end{algorithm}

    First, we show that there is an upper-bound on the size of $\Q$ for any $(\Q, \T)$ pair that is $\D$-separable. Intuitively, since, each access string in $\Q$, necessarily reaches a different state in any minimal finite automaton for $L_\D$, due to the fact that some string (from $\T$) distinguishes it from any other string in $\Q$, the size of $\Q$ can be at most the number of states of a minimal finite automaton, which is less than $N_{\D}$, the number of states of $\D$.
    \begin{lemma}
    \label{lem:auxres2}
         If the pair \((\Q,\T)\) is \(\D\)-separable, then \(\abs{\Q}\) is at most \(N_{\D}\).
    \end{lemma}
    \begin{proof}
    Let $\D^* = (\Q^*, q_{0}^*, \Sigma, \Lambda, \delta^*, \gamma^*)$ be any minimal FA for $L_D$. Note that $\D^*$ has at most with $N_\D$ states. Suppose \(\abs{\Q}\) is greater than $N_\D$. Then, by pigeonhole principle \cite{papa_optimization}, there are two words $u, v \in \Q$ which access the same state of $\D^*$, that is, \(\delta^*(q_{0}^*,u) = \delta^*(q_{0}^*,v)\).  Then, \(\delta^*(q_{0}^*,uw) = \delta^*(q_{0}^*,vw)\) for any word $w$. This implies that $L_{\D^*}(uw) = L{\D^*}(vw)$, which implies that $L_{\D}(uw) = L{\D}(vw)$.
    Hence, we have  $\Output_\D(uw) = \Output_\D(vw)$. Since, that last equivalence holds for any word $w$, it also holds for any $w$ in $\T$, which implies that $u$ and $v$ are  $\T$-equivalent with respect to $\D$. This contradicts the $\D$-separabilty of $(\Q, \T)$, which requires that for any $u, v \in \Q$, $u, v$ are not $\T$-equivalent.
\end{proof}

The next result states that if $(\Q, \T)$ is not closed, then $\Q$ can be expanded, while keeping $\T$ and the $\D$-separability of $(\Q, \T)$ intact.
Note, however, that from Lemma \ref{lem:auxres2}, there is an upper bound on the size of $\Q$, so, the next Lemma also implies that by expanding $\Q$ at most $N_\D$ times, we can obtain a pair $(\Q, \T)$, that is both closed and separable. Also, the expansion at each step is computable. Line \ref{line:2} of Algorithm \ref{algo:semantics} uses this to compute a closed and separable pair $(\Q, \T)$.
    \begin{lemma}
    \label{lem:auxres3}
         If the pair \((\Q,\T)\) is \(\D\)-separable but not \(\D\)-closed, then there is a  \(q\in \Q\) and \(e\in\Sigma\) such that  \((\Q\cup\{qe\},\T)\) is \(\D\)-separable.
    \end{lemma}
    \begin{proof}
        Since the pair \((\Q,\T)\) is not \(\D\)-closed, we have that there exists \(q\in \Q\) and \(e\in\Sigma\) such that \(qe\) is not \(\T\)-equivalent to \(q'\) for any \(q'\in \Q\). Hence, adding $qe$ to $\Q$ preserves separability.
    \end{proof}

   \begin{lemma}
   \label{lem:auxres4a}
        For every \(\D\)-separable pair \((\Q,\T)\), we can compute a \(\D\)-closed and \(\D\)-separable pair \((\tilde{Q},\T)\), where $\Q \subseteq \tilde{Q}$, in time at most $O(N_\D)$
    \end{lemma}
    \begin{proof}
        If the pair \((\Q,\T)\) is not \(\D\)-closed, then using Lemma \label{lem:auxres3} we can effectively find \(q\in \Q\) and \(e\in\Sigma\), by iterating over the elements of finite sets $\Q$ and $\Sigma$, such that \((\Q\cup\{qe\},\T)\) is \(\D\)-separable, which can be effectively checked.
        If the resulting set is not \(\D\)-closed, we can iteratively expand it until it is \(\D\)-closed. Note that by Lemma \ref{lem:auxres2}, there is a bound on the number of elements that can be added, so we will obtain a \(\D\)-closed and \(\D\)-separable set in at most $N_{\D}$ iterations.
    \end{proof}

    Next, we present the details of the algorithm for expanding $(\Q, \T)$ if the hypothesis automaton is incorrect. We will use a counter-example returned by the equivalence checker  to expand the pair $(\Q, \T)$ such that separability is still maintained. This will be again followed by a closure operation to obtain the next hypothesis automaton, and the loop will continue until a hypothesis automaton whose language is that of $\D$ is found.

    \begin{defn}
    \label{d:counter-example}
    \rm{
        A counter-example for \(\D'\) with respect to $\D$ is an input word \(w\in\Sigma^*\) for which the languages of the two automata have different outputs, that is,  \(\Output_{\D}(w)\neq\Output_{\D'}(w)\).
    }
    \end{defn}

    \begin{lemma}
    \label{lem:auxres4}
        Suppose that the pair \((\Q,\T)\) is \(\D\)-separable and \(\D\)-closed, and \(\D'\) be the corresponding hypothesis FA. Given a counter-example \(w\)  for $\D'$ with respect to $\D$, we can compute \(q\in\Sigma^*\setminus \Q\) and \(t\in\Sigma^*\) such that the pair \((\Q\cup\{q\},\T\cup\{t\})\) is \(\D\)-separable using at most $O(\log(\abs{w}))$ IO-generator queries.
    \end{lemma}
    \begin{proof}
         Let $\D' = (\Q', q'_0, \Sigma, \Lambda, \delta', \gamma')$ be the hypothesis automaton constructed from $(\Q, \T)$ using Definition \ref{d:hypo_automat}. Let us define $q'_i$ to be the state reached in $\D'$ after reading $w[1 \cdots i]$, that is, $q'_i = \delta'(q'_0, w[1 \cdots i])$.
         Consider the sequence $O_i = \Output_\D(q'_iw[i+1 \cdots n])$, for $i = 0, \cdots, n$.
         Note that $O_0 = \Output_\D(q'_0 w[1 \cdots n]) = \Output_\D(\epsilon w) = \Output_\D(w)$.
         $O_n = \Output_\D(q'_n w[n+1 \cdots n]) = \Output_\D(q'_n \epsilon) = \Output_\D(q'_n) = \gamma'(q'_n)$ (from Definition \ref{d:hypo_automat}).
         Hence, $O_n = \gamma'(q'_n) = \gamma'(\delta'(q'_0, w[1 \cdots n])) = \gamma'(\delta'(q'_0, w)) = \last(\gamma'(\run_{\D'}(q'_0, w))) = \Output_{\D'}(w)$.
         We obtain that $O_0 \not= O_n$, since, $O_0 = \Output_\D(w)$ and $O_n = \Output_{\D'}(w)$, and $\Output_\D(w) \not= \Output_{\D'}(w)$, since, $w$ is a counter-example of $\D'$ with respect to $\D$.

         We can compute any $O_i$ with a constant number of IO-generator queries, as given by Algorithm \ref{algo:output_comp}.
         We know that $O_0 \not= O_n$, hence, we can perform a binary search to find an $i$ such that $O_i \not= O_{i+1}$. More precisely, given a range $[j \cdots k]$ for which we know $O_j \not= O_k$, we check if $O_l = O_{l+1}$ for the center element of the range $[j \cdots k]$. If they are not equal, we have found an index $i$ that we want. If they are equal, then at least one of the ranges $[j \cdots l]$ or $[l \cdots k]$ will be such that $O_j \not= O_l$ or $O_l \not= O_k$, respectively, and we can continue the search in this smaller range of half the size.
         Hence, in at most $O(\log(\abs{w}))$ IO-generator queries, we can find an index $i$ such that $O_i \not= O_{i+1}$.

         Next, we claim that $\tilde{\Q} = \Q \cup \{q'_iw[i+1 \cdots i+1]\}$ and $\tilde{T} = T\cup\{w[i+2 \cdots n]\}$ are such that $q'_iw[i+1 \cdots i+1] \not\in \Q$ and $(\tilde{Q}, \tilde{T})$ is \(\D\)-separable, that is, we have strictly expanded $\Q$ and we still maintain separability.

          From the definition of $q'_j$'s and the determinism of the finite automaton, we know that $\delta'(q'_i, w_{i+1}) = q'_{i+1}$, where we use $x_i$ to denote $x_{[i \cdots i]}$.
          From the construction of $\D'$ based on $(\Q, \T)$, recall that $q'_i$ and $q'_{i+1}$ are in $\Q$, and $q'_i w_{i+1} \equiv_T q'_{i+1}$.
          Suppose $q'_i w_{i+1} \in \Q$. Since, no two distinct words in $\Q$ are equivalent, we have $q'_i w_{i+1} = q'_{i+1}$.
          But then, $O_i = \Output_\D(q'_iw[i+1 \cdots n])= \Output_\D(q'_iw_{i+1} w[i+2 \cdots n]) = \Output_\D(q'_{i+1}w[i+2 \cdots n]) = O_{i+1}$ contradicting the choice of $i$.
          Hence, we can conclude that $q'_i w_{i+1} \not\in \Q$.

          We need to show that $(\tilde{Q}, \tilde{T})$ is \(\D\)-separable.
          Observe that $(\Q, \T)$ is \(\D\)-separable implies that any two words in $\Q$ are distinguishable using $\T$.
          Further, $q'_i w_{i+1} \equiv_T q'_{i+1}$, since $q'_{i+1} \in \Q$, separability also implies that $q'_i w_{i+1}$ is distinguishable from every $q \not= q'_{i+1}$ in $\Q$.
          So, we just need to show $q'_i w_{i+1}$ is distinguishable from $q'_{i+1}$ using $\tilde{\T}$.
          In fact, $w[i+2 \cdots n] \in \tilde{\T}$ distinguishes $q'_i w_{i+1}$ from $q'_{i+1}$, since, $\Output_\D(q'_i w_{i+1}w[i+2 \cdots n]) \not= \Output_\D(q'_{i+1} w[i+2 \cdots n])$ from the choice of $i$ such that $O_i \not= O_{i+1}$.
   \end{proof}

Next, we state the correctness of the finite automaton learning algorithm.

\begin{theorem}
\label{thm:main}
Algorithm \ref{algo:semantics} always terminates and outputs a finite automaton whose language is $L_\D$.
\end{theorem}
\begin{proof}
Correctness of the algorithm is straight-forward, since, it only outputs $\D'$ that passes the equivalence query with respect to $\D$.
Termination follows from the fact that each of the steps in the algorithm can be effectively computed using the Algorithm \ref{algo:output_comp} for computing $\Output_\D$ and counter-example generator for $L_\D$.
$\Q$ is strictly expanded in each iteration, and there is bound on the size of $\Q$.
\end{proof}

\begin{remark}
\label{rem:change_L*}
\rm{
   Our algorithm is similar to Angluin's algorithm, however, the technical development is performed using the notion of $\Output$ that generalizes two labels to multiple labels, and $\Output$ can be computed using IO-generator queries for our subclass of linear switched systems.
   }
\end{remark}

    \begin{example}
    \label{ex:semantics}
    \rm{
        Consider the switched system described in Example \eqref{ex:swsys}. We apply Algorithm \ref{algo:semantics} to learn a FA \(\D'\) that accepts the language, \(L_{\D}\). The Learner performs the following set of tasks:
        \begin{enumerate}[label = \arabic*., leftmargin = *]
        \item Set \(\Q = \T = \{\varepsilon\}\).
        \item Apply Algorithm \ref{algo:output_comp} to all \(w\in\{\varepsilon,e_1,e_2\}\). It is observed that \((\Q,\T)\) is \(\D\)-separable but not \(\D\)-closed. Indeed, \(\Output_{\D}(\varepsilon e_1 \varepsilon)\neq \Output(\varepsilon\varepsilon)\). Update \(\Q = \{\varepsilon,e_{1}\}\).
        \item Apply Algorithm \ref{algo:output_comp} to all \(w\in\{e_1e_1,e_1e_2\}\). It is observed that \((\Q,\T)\) is \(\D\)-separable but not \(\D\)-closed. Indeed, \(\Output_{\D}(\varepsilon e_2 \varepsilon)\neq \Output(\varepsilon\varepsilon)\) and \(\Output_{\D}(\varepsilon e_2 \varepsilon)\neq \Output(e_1\varepsilon)\). Update \(\Q = \{\varepsilon,e_{1},e_{2}\}\).
        \item Apply Algorithm \ref{algo:output_comp} to all \(w\in\{e_2 e_1,e_2 e_2\}\). It is observed that \((\Q,\T)\) is \(\D\)-separable and \(\D\)-closed. Construct the hypothesis FA \(\D'\) shown in Figure \ref{fig:hypo-automat1}. Checking for correctness of \(\D'\) with the counter-example generator, yields a counter-example \(w = e_1 e_2 e_2\). Update \(\Q = \{\varepsilon,e_{1},e_{2},e_{1}e_{2}\}\) and \(\T = \{\varepsilon, e_{2}\}\).
        \item Apply Algorithm \ref{algo:output_comp} to all \(w\in\{e_1 e_1 e_2, e_1 e_2 e_2, e_1 e_2 e_1, e_2 e_1 e_2, e_2 e_2 e_2, e_1 e_2 e_1 e_2, e_1 e_2 e_2 e_2\}\). It is observed that \((\Q,\T)\) is \(\D\)-separable and \(\D\)-closed. Construct the hypothesis FA \(\D'\) shown in Figure \ref{fig:hypo-automat2}. Checking for correctness of \(\D'\) does not yield a counter-example.
    \end{enumerate}
    We conclude that \(\D'\) obtained in Step 5. accepts the language, \(L_{\D}\).
    \begin{figure}[htbp]
    \centering
    \scalebox{0.8}{
        \begin{tikzpicture}[every path/.style={>=latex},base node/.style={draw,circle}]
            \node[base node] (a) at (0,0)  { $A_{1}$ };
            \node[base node] (b) at (2,0)  { $A_{2}$ };
            \node[base node] (c) at (0,-2)  { $A_{3}$ };

            \draw[->] (a) edge[bend left] (b);
            \draw[->] (b) edge (a);

            \draw[->] (a) edge[bend left] (c);
            \draw[->] (c) edge (a);

            \node (s) at (1,0.6) {$e_{2}$};
            \node (s) at (1,-0.2) {$e_{1},e_{2}$};
            \node (s) at (0.6,-1) {$e_{1}$};
            \node (s) at (-0.5,-1) {$e_{1},e_{2}$};

            \draw[->] (-1,0) -- (a);
        \end{tikzpicture}
        }
        \caption{Hypothesis FA \(\D'\) in Step 4. of Example \ref{ex:semantics}} \label{fig:hypo-automat1}
    \end{figure}
      \begin{figure}[htbp]
    \centering
        \scalebox{0.8}{
        \begin{tikzpicture}[every path/.style={>=latex},base node/.style={draw,circle}]
            \node[base node] (a) at (-2,2)  { $A_{1}$ };
            \node[base node] (b) at (0,2)  { $A_{2}$ };
            \node[base node] (d) at (-2,-2)  { $A_{3}$ };
            \node[base node] (c) at (0,-2)  { $A_{2}$ };

            \draw[->] (a) edge (b);
            \draw[->] (b) edge[bend right] (a);
            \draw[->] (b) edge (c);
            \draw[->] (c) edge[bend right] (b);
            \draw[->] (c) edge (d);
            \draw[->] (d) edge[bend right] (c);
            \draw[->] (a) edge (d);
            \draw[->] (d) edge[bend left] (a);

             \node (s) at (1,0.2) {$e_{1}$};
             \node (s) at (-0.2,0.2) {$e_{1}$};
             \node (s) at (-1.8,0.2) {$e_{1}$};
             \node (s) at (-2.9,0.2) {$e_{1}$};

             \node (s) at (-1,1.8) {$e_{2}$};
             \node (s) at (-1,2.6) {$e_{2}$};
             \node (s) at (-1,-2.6) {$e_{2}$};
             \node (s) at (-1,-1.8) {$e_{2}$};

            \draw[->] (-3,2) -- (a);
        \end{tikzpicture}
        }
        \caption{Hypothesis FA \(\D'\) in Step 5. of Example \ref{ex:semantics}} \label{fig:hypo-automat2}
    \end{figure}
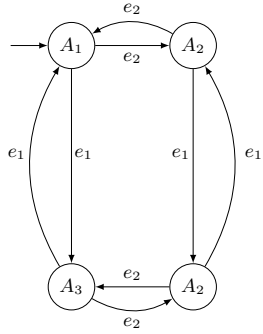
    }
    \end{example}

    To wrap up, let us discuss the problem of learning the switched system.
    Given a switched system $\D$ with $d$ and $\Sigma$ known, Algorithm \ref{algo:semantics} outputs a switched system $\D'$ whose executions coincide with that of $\D$.
    \begin{corollary}
    \label{t:mainres}
        Algorithm \ref{algo:semantics} outputs a switched system $\D'$ such that $\exec_\D(x, w) = \exec_{\D'}(x,w)$ for every $x \in \R^d$ and $w \in \Sigma^*$.
    \end{corollary}
    \begin{proof}
    This follows immediate from Theorem \ref{thm:main}, where we established that $L_\D = L_{\D'}$, and the fact that $\exec(x, w)$ depends only on $x$ and $L(w)$.
    \end{proof}

\begin{remark}
\label{rem:L*_compa}
\rm{
    Given the set of events, \(\Sigma\), the dimension of the subsystems matrices, \(d\), and the IO-generator and counter-example generators, Algorithm \ref{algo:semantics} learns an FA that accepts the semantics of the underlying FA of the unknown switched system under consideration. The learning technique employed in Algorithm \ref{algo:semantics} is an extension of the \(L^*\)-algorithm. In the \(L^*\)-algorithm, the Learner learns an event-deterministic unlabelled finite automaton that accepts a certain language \(L\), with the aid of an Oracle called the \emph{minimally adequate teacher} (MAT). An automaton \(\A\) under consideration in \cite{Angluin1987} is a tuple \((P,p_{0},\Gamma,F,\mu)\), where \(P\) is a finite set of nodes, \(p_{0}\in P\) is the unique initial node, \(\Gamma\) is a finite set of alphabets, \(F\subseteq P\) is a finite set of accepting (or final) nodes, and \(\mu:P\times\Gamma\to P\) is the node transition function. The language of \(\A\) is the set of all finite words (strings of alphabets) such that the automaton reaches a final node on reading them, i.e., a word \(w = w_{1}w_{2}\cdots w_{m}\), \(w_{k}\in\Gamma\), \(k=1,2,\ldots,m\), belongs to the language of \(\A\), if \(\mu(\cdots(\mu(\mu(p_{0},w_{1}),w_{2}),\cdots,w_{m})\in F\). The MAT knows \(L\) and answers two types of queries by the Learner: \emph{membership queries}, i.e., whether or not a given word belongs to \(L\), and \emph{equivalence queries}, i.e., whether a hypothesis automaton specified by the Learner is correct or not. If the language of the hypothesis automaton differs from \(L\), then the MAT responds to an equivalence query with a counter-example, which is a word that is misclassified by the hypothesis automaton. The class of automata considered in this paper differs structurally from the class of automata considered in \cite{Angluin1987} in the following ways:
    \begin{enumerate}[label = (\alph*), leftmargin = *]
        \item \(\D\) has \(0\)-many accepting nodes, and
        \item the nodes of \(\D\) are labelled with matrices.
    \end{enumerate}
    In Algorithm \ref{algo:semantics} we modify the \(L^*\)-algorithm to cater to learning of \(\D\). At this point, it is important to highlight that throughout this paper we have employed notations, terminologies and concepts from the version of \(L^*\)-algorithm presented in \cite{Tut_L*}. Loosely speaking, the IO-generator and counter-example generator together play the role of a MAT. Indeed, the IP-generator provide the Learner with finite traces of state trajectories of the unknown \(\D\) under consideration. The Learner then uses this information to compute \(\Output_{\D}(w)\) for \(w\in\Sigma^*\) that satisfy certain conditions. In addition, the the counter-example generator facilitates checking correctness of an FA hypothesized by the Learner.
    }
\end{remark}
     
    \begin{remark}
    \label{rem:label_compa}
    \rm{
    Earlier in \cite{Angluin2009} the role of labels on the nodes of an automaton were employed in the setting of the \(L^*\)-algorithm to aid the learning process. The authors allow the MAT to make an automaton easier to learn by adding binary scalar labels to its nodes, either carefully or randomly chosen. When the Learner performs a membership query for a string, then she not only receives whether it is accepting or not, but also is provided with the label of the node that the automaton reaches on its application. It is shown that if the node labels are distinct, then the learning process becomes easier, and if all the node labels are same, then the learning may require an exponential number of queries. The above set of observations does not extend readily to our setting due to the structural difference of our FA's with the class of automata considered in the \(L^*\)-algorithm described above. Indeed, our FA's do not have final nodes and labelling of the nodes with matrices is governed by the underlying switching rules of the system under consideration. Beyond identification of switched systems, our learning algorithm is applicable to the general setting of learning deterministic finite automaton with \(0\)-many final nodes and all nodes labelled with full-rank matrices.
    }
    \end{remark}

    \begin{remark}
    \label{rem:aids}
    \rm{
        Notice that the IO-generator and the counter-example generator can be thought of as a simulation model of the unknown switched system, \(\D\). In modern industrial setups, simulation is of prior importance. Such models for complex systems are often provided by the system manufacturers. The mathematical models of the system components and the constraints on their operations underlying the simulation model are typically not made known explicitly to the user, but the model can be used to study the system behaviour with respect to various sets of inputs prior to their application to the actual system. Given a simulation model that allows the set of operations by the user required for our setting, the Learner can generate finite traces of trajectories of a switched system with respect to sets of initial states and sequences of events. This serves for the purpose of Algorithm \ref{algo:output_comp}. For the generation of a counter-example, the Learner can apply sequences of events of increasing length (up to a sufficiently large number) and match the labels of the nodes reached on \(\D\) and \(\D'\).
    }
    \end{remark}
    
    We now move on to a set of experiments conducted to demonstrate the effectiveness and performance of our learning algorithm. 
	\section{Numerical experiments}
\label{s:num_ex}
    We first describe the implementation of our learning algorithm on a MATLAB R2020a platform. We will then demonstrate the performance of our algorithm on a set of examples.
    
    A primary requirement for the implementation of the proposed algorithm is the design of an IO-generator and a Counter-example
    generator. Towards this end, we construct a MATLAB routine \texttt{fa-oracle.m} that knows \(\D\) and can perform the following task: accept an input \((x,w)\) and output \(\exec_{\D}(x,w)\). The Learner routine \texttt{fa-learn.m} uses \texttt{fa-oracle.m} as both an IO-generator and a Counter-example generator. Using \texttt{fa-oracle.m} as an IO-generator is straightforward. Towards using it as a Counter-example generator, \texttt{fa-learn.m} performs the following tasks: (a) it fixes a hypothesis automaton \(\D'\), (b) chooses a large number \(L\), (c) computes \(\Output_{\D}(w)\) for all possible \(w\) of increasing length, one at a time, by means of Algorithm \ref{algo:output_comp} and the routine \texttt{fa-oracle.m}, and (d) matches \(\Output_{\D}(w)\) with \(\Output_{\D'}(w)\). This procedure is continued until either a counter-example \(w\) is obtained or all \(w\) of length \(i=1,2,\ldots,L\) are exhausted.
    
    We now present a set of experiments conducted in the above setting. The hardware platform used is an Intel 17-8550U, 8GB RAM, 1TB HDD machine with Windows 10 Operating System.
   
   Our first example is motivated by a practical application often encountered in systems with variable structures and/or multiple controllers.
\begin{example}
\label{ex:num_ex1}
\rm{
    Consider a linear plant with \(3\) modes of operations. Under a healthy condition, the plant follows a pre-specified schedule for mode selection. Whenever a fault occurs, the plant continues to dwell on the current mode of operation until the fault is cleared.

    The setting described above can be expressed as an internally event-driven switched system for which \(\D\) is as shown in Figure \ref{fig:fa2}. Let \(A_{1} = \pmat{0.2 & 0.4 & 0.8\\0.3 & 0.6 & 0.9\\0.5 & 1.5 & 1.5}\), \(A_{2} = \pmat{-1 & 0.1 & 0.2\\0.3 & -1 & 0.4\\0.5 & 0.6 & -1}\), \(A_{3} = \pmat{-0.1 & -0.2 & 0.3\\-0.1 & -0.4 & 0.6\\0.8 & 0.7 & -0.6}\).

    Notice that the matrices \(A_1, A_2\) and \(A_3\) are full-rank. The following steps are carried out:
    \begin{enumerate}[label = \arabic*., leftmargin = *]
        \item Set \(\Q = \T = \{\varepsilon\}\).
        \item Apply Algorithm \ref{algo:output_comp} to all \(w\in\{\varepsilon,\f,\g\}\). It is observed that \((\Q,\T)\) is \(\D\)-separable but not \(\D\)-closed. Indeed, \(\Output_{\D}(\varepsilon\cdot\g) \neq \Output_{\D}(\varepsilon)\). Update \(\Q = \{\varepsilon,\g\}\).
        \item Apply Algorithm \ref{algo:output_comp} to all \(w\in\{\g\cdot\f,\g\cdot\g\}\). It is observed that \((\Q,\T)\) is \(\D\)-separable and \(\D\)-closed. Construct the hypothesis FA \(\D'\) shown in Figure \ref{fig:hypo-automat3}. Checking for correctness of \(\D'\) with the counter-example generator, yields a counter-example \(w = \g\cdot\g\cdot \g\). Update \(\Q = \{\varepsilon,\g,\g\cdot\g\}\) and \(\T = \{\varepsilon,\g\}\).
        \item Apply Algorithm \ref{algo:output_comp} to all \(w\in\{\g\cdot\f,\g\cdot\g,\f\cdot\g,\g\cdot\f\cdot\g,\g\cdot\g\cdot\g,\g\cdot\g\cdot\f,\g\cdot\g\cdot\f\cdot\g,\g\cdot\g\cdot\g\cdot\f\}\). It is observed that \((\Q,\T)\) is \(\D\)-separable but not \(\D\)-closed. Indeed,
        \(\Output_{\D}(\g\cdot\g\cdot\g) \neq \Output_{\D}(\varepsilon)\), \(\Output_{\D}(\g\cdot\g\cdot\g) \neq \Output_{\D}(\g)\), \(\Output_{\D}(\g\cdot\g\cdot\g) \neq \Output_{\D}(\varepsilon\cdot\varepsilon)\). Update \(\Q = \{\varepsilon,\g,\g\cdot\g,\g\cdot\g\cdot\g\}\).
        \item Apply Algorithm \ref{algo:output_comp} to all \(w\in\{\g\cdot\g\cdot\g\cdot\f,\g\cdot\g\cdot\g\cdot\f\cdot\g,\g\cdot\g\cdot\g\cdot\g,\g\cdot\g\cdot\g\cdot\g\cdot\g\}\). It is observed that \((\Q,\T)\) is \(\D\)-separable and \(\D\)-closed. Construct the hypothesis FA \(\D'\) shown in Figure \ref{fig:hypo-automat4}. Checking for correctness of \(\D'\) with the counter-example generator does not yield a counterexample.
    \end{enumerate}
    We conclude that \(\D'\) is a FA that accepts the language, \(L_{\D}\).
\begin{figure}[htbp]
    \centering
        \scalebox{0.8}{
        \begin{tikzpicture}[every path/.style={>=latex},base node/.style={draw,circle}]
            \node[base node] (a) at (-4,0)  { $A_{1}$ };
            \node[base node] (b) at (-2,0)  { $A_{2}$ };
            \node[base node] (c) at (0,0)  { $A_{2}$ };
            \node[base node] (d) at (2,0)  { $A_{3}$ };

             \draw[->] (a) edge[loop above] (a);
             \draw[->] (b) edge[loop above] (b);
             \draw[->] (c) edge[loop above] (c);
             \draw[->] (d) edge[loop above] (d);
            \node (s) at (0,1.1) {$\f$};
            \node (s) at (-4,1.1) {$\f$};
            \node (s) at (-2,1.1) {$\f$};
            \node (s) at (2,1.1) {$\f$};
            \draw[->] (a) edge (b);
            \draw[->] (b) edge (c);
            \draw[->] (c) edge (d);
            \node (s) at (1,0.2) {$\g$};
            \node (s) at (-1,0.2) {$\g$};
            \node (s) at (-3,0.2) {$\g$};
            \draw[->] (d) edge[bend left] (a);
            \node (s) at (-1,-0.8) {$\g$};
            \draw[->] (-5,0) -- (a);
        \end{tikzpicture}
        }
        \caption{FA for Example \ref{ex:num_ex1}} \label{fig:fa2}
    \end{figure}
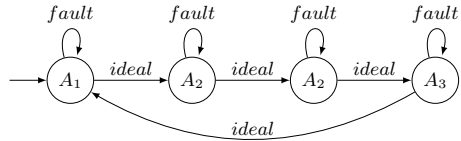
    \begin{figure}[htbp]
    \centering
    \scalebox{0.8}{
        \begin{tikzpicture}[every path/.style={>=latex},base node/.style={draw,circle}]
            \node[base node] (a) at (-2,0)  { $A_{1}$ };
            \node[base node] (b) at (0,0)  { $A_{2}$ };

             \draw[->] (a) edge[loop above] (a);
             \draw[->] (b) edge[loop above] (b);
            \node (s) at (0,1.1) {$\f,\g$};
            \node (s) at (-2,1.1) {$\f$};
            \draw[->] (a) edge (b);
            \node (s) at (-1,0.2) {$\g$};
            \draw[->] (-3,0) -- (a);
        \end{tikzpicture}
        }
        \caption{Hypothesis automaton \(\D'\) in Step 3. of Example \ref{ex:num_ex1}} \label{fig:hypo-automat3}
    \end{figure}
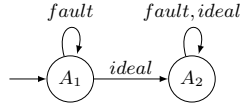
    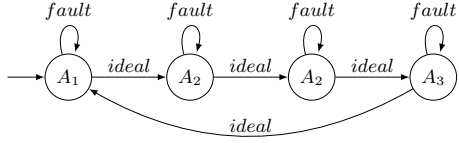
\begin{figure}[htbp]
    \centering
        \scalebox{0.8}{
        \begin{tikzpicture}[every path/.style={>=latex},base node/.style={draw,circle}]
            \node[base node] (a) at (-4,0)  { $A_{1}$ };
            \node[base node] (b) at (-2,0)  { $A_{2}$ };
            \node[base node] (c) at (0,0)  { $A_{2}$ };
            \node[base node] (d) at (2,0)  { $A_{3}$ };

             \draw[->] (a) edge[loop above] (a);
             \draw[->] (b) edge[loop above] (b);
             \draw[->] (c) edge[loop above] (c);
             \draw[->] (d) edge[loop above] (d);
            \node (s) at (0,1.1) {$\f$};
            \node (s) at (-4,1.1) {$\f$};
            \node (s) at (-2,1.1) {$\f$};
            \node (s) at (2,1.1) {$\f$};
            \draw[->] (a) edge (b);
            \draw[->] (b) edge (c);
            \draw[->] (c) edge (d);
            \node (s) at (1,0.2) {$\g$};
            \node (s) at (-1,0.2) {$\g$};
            \node (s) at (-3,0.2) {$\g$};
            \draw[->] (d) edge[bend left] (a);
            \node (s) at (-1,-0.8) {$\g$};
            \draw[->] (-5,0) -- (a);
        \end{tikzpicture}
        }
        \caption{Hypothesis automaton \(\D'\) in Step 5. of Example \ref{ex:num_ex1}} \label{fig:hypo-automat4}
    \end{figure}
    }
    \end{example}

    \label{rem:exmpl_compa}
    \rm{
        The FA considered in Example \ref{ex:num_ex1} resembles the automata used to implement \(L^*\)-algorithm in \cite[\S 2]{Tut_L*} without final nodes and with node labels. We note that the total number of membership queries and equivalence queries required for the learning task in \cite[Section 2]{Tut_L*} matches the total number of calls to Algorithm \ref{algo:output_comp} and the check for correctness of hypothesis FA in our setting.
        }

    We next conduct an experiment to verify scalability of our learning algorithm.
    \begin{example}
    \label{ex:num_ex2}
    \rm{
        We choose two benchmark examples described in \cite[\S4.4]{Neider2019}.\footnote{The benchmark examples under consideration are for Moore Machines, and does not involve matrices. We, therefore, choose the dimension of the subsystems matrices to cater to our purpose.} The following procedure is executed in each case:
        \begin{enumerate}[label = \arabic*), leftmargin = *]
            \item Construction of a switched system:
            \begin{enumerate}[label = \Alph*), leftmargin = *]
                \item We specify the number of nodes, \(\abs{\Q}\), the number of events, \(\abs{\Sigma}\), the number of labels, \(\abs{\Lambda}\), and the dimension of the subsystems matrices, \(d\).
                \item We randomly generate a FA \(\D\) that obeys the above features. The following are ensured: a) \(\D\) is complete in the sense that there is a valid transition corresponding to every pair of node and event, and (b) \(\D\) has a unique initial node.
                \item We randomly generate the matrices, \(A_j\in\R^{d\times d}\), \(j=1,2,\ldots,N\), where \(N\) is the number of labels used in \(\D\), \(N\leq\abs{\Lambda}\). It is ensured that each \(A_j\), \(j\in\{1,2,\ldots,N\}\), is full-rank.
            \end{enumerate}
            \item Learning the switched system constructed above: We employ the MATLAB routines \texttt{fa-oracle.m} and \texttt{fa-learn.m} as described above, to learn a switched system generated in Step 1). 
        \end{enumerate}
        We note the execution times of our algorithm in Table \ref{tab:algo_stat}.
        \begin{table}[htbp]
        \begin{center}
        \begin{tabular}{c c c c c}
            \hline\\
             \(\abs{\Q}\) & \(\abs{\Sigma}\) & \(\abs{\Lambda}\) & \(d\) & Execution time of Algorithm \ref{algo:output_comp}\\
             \hline\\
             \(1000\) & \(19\) & \(19\) & \(100\) & \(23743\) seconds (\(\approx 7\) hours)\\
             \hline\\
             \(2000\) & \(9\) & \(9\) & \(100\) & \(41489\) seconds (\(\approx 12\) hours)\\
             \hline
        \end{tabular}
        \caption{Data for Example \ref{ex:num_ex2}}\label{tab:algo_stat}
        \end{center}
        \end{table}
        }
    \end{example}
%
%
    \rm{
        It is observed that in each case the automaton is learnt correctly. However, with the increase in the size of the target automaton the number of queries (and hence the execution time) increases. On the one hand, since the learning procedure is offline, a longer time of execution for large-scale settings, as observed in Example \ref{ex:num_ex2}, is acceptable. However, it is of interest to derive mathematical guarantees on the performance of our algorithm with respect to the distribution of the elements of \(\Lambda\) on the elements of \(\Q\) (along the lines of \cite{Angluin2009}). We identify this problem as a direction for future work.
        }

    \section{Conclusion}
\label{s:concln}
    In this paper, we presented a learning algorithm for the identification of event-driven switched linear systems. We demonstrated our algorithm on various examples. Our future research directions include the design of active learning techniques for large-scale switched systems whose subsystems dynamics are not restricted to be linear structures and/or the available state-trajectories are noisy.


\end{document}